\documentclass[submission,copyright,creativecommons]{eptcs}
\providecommand{\event}{ICE 2014}

\usepackage{amsmath}
\usepackage{amssymb}
\usepackage{times}
\usepackage{txfonts}
\usepackage{stmaryrd}
\usepackage{url}
\usepackage{amsfonts}
\usepackage{code}
\usepackage{mathrsfs}
\DeclareMathAlphabet{\mathpzc}{OT1}{pzc}{m}{it}
\let\mathpzc\mathscr
\let\mathpzc\mathcal
\def\BNF{\ \  | \ \  }
\def\bondi{{\bf bondi}}
\def\ltrans{[\![}
\def\rtrans{]\!]}

\newtheorem{theorem}{Theorem}[section]
\newtheorem{lemma}[theorem]{Lemma}

\newenvironment{proof}[1][Proof]{\begin{trivlist}
\item[\hskip \labelsep {\bfseries #1}]}{\end{trivlist}}

\usepackage{prooftree}
\usepackage{macros}

\title{An Intensional Concurrent Faithful Encoding of Turing Machines}

\author{Thomas Given-Wilson
\institute{INRIA, Paris, France
\footnote{This work has been supported by the project ANR-12-IS02-001 PACE.}
}
\email{thomas.given-wilson@inria.fr}
}

\def\titlerunning{An Intensional Concurrent Faithful Encoding of Turing Machines}
\def\authorrunning{T. Given-Wilson}

\makeatletter
\def \rightarrowfill{\m@th\mathord{\smash-}\mkern-6mu%
  \cleaders\hbox{$\mkern-2mu\mathord{\smash-}\mkern-2mu$}\hfill
  \mkern-6mu\mathord\to}
\makeatother
\makeatletter
\def \Rightarrowfill{\m@th\mathord{\smash=}\mkern-6mu%
  \cleaders\hbox{$\mkern-2mu\mathord{\smash=}\mkern-2mu$}\hfill
  \mkern-6mu\mathord\Rightarrow}
\makeatother
\def \overstackrel#1#2{\mathrel{\mathop{#1}\limits^{#2}}}

\renewcommand{\iap}[1]{#1\to}
\renewcommand{\oap}[1]{\overline{#1}\to}
\newcommand{\oan}[1]{\overline{#1}}
\newcommand{\tmach}[1]{\tuple{\tuple{#1}}}
\newcommand{\tape}[1]{[#1]}

\begin{document}
\makeatactive

\maketitle  

\begin{abstract}
The benchmark for computation is typically given as Turing computability;
the ability for a computation to be performed by a Turing Machine.
Many languages exploit (indirect) encodings of Turing Machines to demonstrate
their ability to support arbitrary computation.
However, these encodings are usually by simulating the entire Turing Machine
within the language, or by encoding a language that does an encoding or
simulation itself.
This second category is typical for process calculi that show an encoding of
$\l$-calculus (often with restrictions) that in turn simulates a Turing Machine.
Such approaches lead to indirect encodings of Turing Machines that are complex,
unclear, and only weakly equivalent after computation.
This paper presents an approach to encoding Turing Machines into intensional
process calculi that is faithful, reduction preserving, and structurally
equivalent.
The encoding is demonstrated in a simple asymmetric concurrent pattern calculus
before generalised to simplify infinite terms, and to show encodings into
Concurrent Pattern Calculus and Psi Calculi.
\end{abstract}

\section{Introduction}
\label{sec:intro}

The benchmark for computation is typically given as Turing computability
\cite{zbMATH02522537,Kleene52,McCarthy:1960:RFS:367177.367199,boolos2007computability};
the ability for a computation to be performed by a Turing Machine
\cite{turing36}.
This choice of benchmark is also widely supported by various appeals to
calculation of a ``computable function'' or ``decidable predicate'' or
``recursive function'' by a Turing Machine
\cite{zbMATH02522537,Kleene52,Curry58combinatorylogic,Curry72combinatorylogic,jay2011}.
Indeed, since Turing \cite{zbMATH02522537}, Kleene \cite{Kleene52}, Curry
\cite{Curry58combinatorylogic,Curry72combinatorylogic}
and others showed that Turing Machines can encode $\l$-calculus, general recursive functions, 
and combinatory logic \cite{685558}, respectively,
any language that can encode any of these can be considered to be able to
do computation.
However, these rely upon their encoding of Turing Machines which typically
involve the simulation of a Turing Machine within the other language.

The typical simulation of a Turing Machine, say in $\l$-calculus, is to
represent the tape as a list, and the symbols by natural numbers using G\"odelisation.
The operations of the Turing Machine are then handled by a function that
can operate over the list (encoded tape) and compare the numbers using their
Church encodings \cite{Barendregt85}.
While such encodings preserve computation they have some weaknesses.
The encoded computation takes many more reductions to produce the same operations;
recognising a symbol requires a predecessor function, testing for zero, and then
switching on numbers to determine the next symbol to write, all before
reconstructing the list (encoded tape).
Such encodings are not very clear; the representation of a symbol $s$ may be
mapped to some number $i$ that is then represented as a function that is the
$i$th iterator.
These kinds of encodings are also metamathematical \cite{Tarski56logic} in nature and
so are always at least one level of mathematics away from the computation itself,
which can lead to misunderstandings about the true expressiveness of a
language \cite{jay2011}.

Process calculi are often considered to generalise the sequential (non-concurrent)
computation of $\l$-calculus by some form of encoding
\cite{90426,Berry:1989:CAM:96709.96717,Milner:1992:CMP:162037.162038,Milner:1992:CMP:162037.162039,DBLP:conf/fossacs/CardelliG98,DBLP:books/daglib/0098267,
705654}.
These encodings again have weaknesses such as adding reductions steps,
lacking clarity, or even limiting reductions strategies (such as in
Milner's encoding of $\l$-calculus into $\pi$-calculus \cite{90426}, which is
then built upon by those who use encoding $\pi$-calculus to capture
computation).
Further, these encodings are often up to some weak behavioural equivalence and
can create many dead processes as a side effect.
Thus a Turing Machine can be encoded into $\l$-calculus and then encoded into
$\pi$-calculus and then encoded into another process calculus so that the 
original computation is now buried three levels of meta-operations deep,
with almost no obvious connection to the original Turing Machine,
and only weakly behaviourally equivalent to an encoding of the Turing Machine
after the computation.

This paper attempts to avoid the worst of these various encoding issues by
presenting a straightforward approach to encoding a Turing Machine into any process
calculus that supports {\em intensional} communication
\cite{GivenWilsonPhD,GivenWilsonGorla13}.
Intensionality is the capability for functions or communication primitives
to operate based upon the internal structure of terms that are arguments
or being communicated, respectively \cite{jay2011,GivenWilsonPhD}.
Some recent process calculi support intensionality, in particular
Concurrent Pattern Calculus \cite{GivenWilsonGorlaJay10,givenwilson:hal-00987578} and
Psi Calculi \cite{BJPV11}.
The presentation here will use a simplified {\em asymmetric concurrent
pattern calculus} (ACPC) to detail the translation as clearly as possible.

The intensionality of ACPC is an advanced form of pattern-matching
that allows {\em compound} structures of the form $s\bullet t$ to be bound to a single
name, or to have their structure and components be matched in communication.
For example, consider the following processes:
\begin{equation*}
P \ \define\  \oap {a\bullet b}\zero\qquad\qquad
Q \ \define\  \iap {\l x\bullet \l y} Q'\qquad\qquad
R \ \define\  \iap {\l z} R'\qquad\qquad
S \ \define\  \iap {a\bullet b}S'
\end{equation*}
where $P$ is an output of the compound $a\bullet b$.
The inputs of $Q$ and $R$ have binding names of the form $\l x$ in their
patterns $\l x\bullet \l y$ and $\l z$, respectively.
The input of $S$ tests the names $a$ and $b$ for equality and performs no binding.
These process can be combined to form three possible reductions:
\begin{eqnarray*}
P\bnf Q &\redar& \{a/x,b/y\}Q'\\
P\bnf R &\redar& \{a\bullet b/z\}R'\\
P\bnf S &\redar& S'\; .
\end{eqnarray*}
The first matches the structure of the output of $P$ with the input of $Q$ and binds
$a$ and $b$ to $x$ and $y$, respectively, in $Q'$.
The second binds the entire output of $P$ to the single name $z$ in $R'$.
The third matches the structure and names of the output of $P$ with the structure
and names of the input of $S$, as they match they interact although no binding
or substitution occurs.

The encoding presented here exploits the ability to represent symbols and structures
into the output of a process, and to test for structure, equality, and to bind in
an input to clearly represent a Turing Machine.
Indeed, the encoding is {\em faithful} in that each operation that is performed
by a Turing Machine yields exactly one reduction in the encoding to the (structurally)
equivalent encoded Turing Machine after the operation.
The key to the elegance of the encoding is to represent the current state $q$ and
tape ${\mathcal T}$ of the Turing Machine by an output of the form
\begin{equation*}
\oap {q\bullet\encode {{\mathcal T}}} \zero
\end{equation*}
where $\encode{{\mathcal T}}$ converts the tape to a convenient format.
The transitions functions of the Turing Machine are then each encoded into
a process of the form
\begin{equation*}
\iap {q_i\bullet p}  \oap {q_j\bullet t} \zero
\end{equation*}
where $q_i$ and $p$ match the current state, current symbol, and the structure of the
tape, and the output exhibits the new state $q_j$ and modified representation of the tape $t$.
These transition functions can then be combined via parallel composition and replication
in a manner that allows for a faithful encoding of a Turing Machine.

The elegance of this encoding can be built upon by considering some variations.
It is straightforward to modify the encoding so that a tape with infinite blank symbols
at the edges can be represented by a finite term in the output.
Both of these encodings can then be easily captured by both Concurrent Pattern Calculus
and Psi Calculi with only minor changes to the encoding and proofs.

There are two limitations for non-intensional process calculi.
First is the inability to match multiple names in a single reduction.
This can be worked around by encodings or match-rules that rely upon structural equivalence,
however at some cost to the elegance of the encoding.
Second proves impossible to fix; the inability for non-intensional calculi to bind
an arbitrarily complex structure to a single name and still access the components.
Thus alternative approaches must be used to encode Turing Machines into, say,
synchronous polyadic $\pi$-calculus, at the cost of faithfulness and easy equivalence
of encodings.

\medskip

The structure of the paper is as follows.
Section~\ref{sec:tm} recalls Turing Machines.
Section~\ref{sec:int} presents intensionality via asymmetric concurrent pattern calculus.
Section~\ref{sec:enc} defines the encoding of Turing Machines in asymmetric concurrent pattern calculus.
Section~\ref{sec:var} considers variations on the encoding, including into published calculi.
Section~\ref{sec:fail} discusses non-faithful encodings into, and limitations of other calculi.
Section~\ref{sec:conc} draws conclusions.

\section{Turing Machines}
\label{sec:tm}

Each Turing Machine is defined by an alphabet, a set of states, a transition function,
and a tape. In addition during operation the current head position and current state
must also be accounted for.
The alphabet ${\mathcal S}$ is the set of symbols $s$ recognised by the Turing Machine and includes
a special {\em blank} symbol $b$.
The set of states ${\mathcal Q}$ is a collection of states $q$ that the Turing Machine can transition
to and from, and includes a {\em start state} $q_0$.
The transition function ${\mathcal F}$ is represented by tuples of the form
$\tuple{q_1,s_2,s_3,d_4,q_5}$ that instructs the machine when the current state is $q_1$ and the
current head position symbol is $s_2$ to write (to the current head position) $s_3$ and
then move the current head position direction $d_4$ (either $L$ for left or $R$ for right),
and change the current state to $q_5$.
The {\em tape} ${\mathcal T}$ is an infinite sequence of cells each of which contains a symbol,
this is denoted by $\tape{\ldots b,s_1,s_2,s_1,b\ldots}$ that indicates an infinite sequence
of blanks, the symbols $s_1$ then $s_2$ then $s_1$ and then an infinite sequence of blanks.
The current head position can be represented by marking the pointed to symbol in bold,
e.g.~the tape $\tape{\ldots b,{\bf s_1},s_2,s_1,b\ldots}$ indicates that the current head position is
the leftmost instance of $s_1$.
Thus the definition of a Turing Machine can be given by
$\tmach{{\mathcal S}, {\mathcal Q}, {\mathcal F}, {\mathcal T}, q_i}$ where
$q_i$ is the current state.

For a Turing Machine $\tmach{{\mathcal S}, {\mathcal Q}, {\mathcal F}, {\mathcal T}, q_i}$
a state $q_i$ is a {\em terminating} state if there are no transitions of the
form $\tuple{q_i,s_1,s_2,d_3,q_j}$.
A Turing Machine is well formed if for every $q_i$ then either $q_i$ is terminating, or
for every symbol $s_j\in{\mathcal S}$ then there exists a transition of the
form $\tuple{q_i,s_j,s_1,d_2,q_3}$ for some $s_1$ and $d_2$ and $q_3$.
The rest of this paper shall only consider well formed Turing Machines
although no results rely upon this.

\subsection*{A Simple Example}

Consider the following simple example of a Turing Machine that accepts numbers
represented in unary and terminates with the current head on $1$ if the number is
even and $b$ if the number is odd.

The alphabet is given by ${\mathcal S}=\{b,1\}$ and the states by
${\mathcal Q}=\{q_0,q_1,q_2,q_3\}$.
The transition function ${\mathcal F}$ is defined as follows:
\begin{equation*}
\begin{array}{c}
\tuple{q_0,b,1,L,q_2}\\
\tuple{q_0,1,b,R,q_1}
\end{array}
\qquad\qquad
\begin{array}{c}
\tuple{q_1,b,b,L,q_3}\\
\tuple{q_1,1,b,R,q_0}
\end{array}
\qquad\qquad
\begin{array}{c}
\tuple{q_2,b,b,R,q_3}\ \ \\
\tuple{q_2,1,b,R,q_3}\; .
\end{array}
\end{equation*}
Observe that $q_3$ is a terminating state and so has not transitions that
begin in that state.

Now consider the Turing Machine given by
$\tmach{{\mathcal S}, {\mathcal Q}, {\mathcal F}, \tape{\ldots,b,{\bf 1},1,1,b,\ldots}, q_0}$, that is
the Turing Machine defined above with the current head position on the first $1$ of the
number three represented in unary.
The computations progress as follows:
\begin{eqnarray*}
\ &\ &\tmach{{\mathcal S}, {\mathcal Q}, {\mathcal F}, \tape{\ldots,b,{\bf 1},1,1,b,\ldots}, q_0}\\
&\redar&\tmach{{\mathcal S}, {\mathcal Q}, {\mathcal F}, \tape{\ldots,b,b,{\bf 1},1,b,\ldots}, q_1}\\
&\redar&\tmach{{\mathcal S}, {\mathcal Q}, {\mathcal F}, \tape{\ldots,b,b,b,{\bf 1},b,\ldots}, q_0}\\
&\redar&\tmach{{\mathcal S}, {\mathcal Q}, {\mathcal F}, \tape{\ldots,b,b,b,b,{\bf b},\ldots}, q_1}\\
&\redar&\tmach{{\mathcal S}, {\mathcal Q}, {\mathcal F}, \tape{\ldots,b,b,b,{\bf b},b,\ldots}, q_3}\; .
\end{eqnarray*}
Since $q_3$ is a terminating state the Turing Machine halts and has the current head position on a
blank as required.
A similar Turing Machine with a tape that represents two would have the following reductions
\begin{eqnarray*}
\ &\ &\tmach{{\mathcal S}, {\mathcal Q}, {\mathcal F}, \tape{\ldots,b,{\bf 1},1,b,\ldots}, q_0}\\
&\redar&\tmach{{\mathcal S}, {\mathcal Q}, {\mathcal F}, \tape{\ldots,b,b,{\bf 1},b,\ldots}, q_1}\\
&\redar&\tmach{{\mathcal S}, {\mathcal Q}, {\mathcal F}, \tape{\ldots,b,b,b,{\bf b},\ldots}, q_0}\\
&\redar&\tmach{{\mathcal S}, {\mathcal Q}, {\mathcal F}, \tape{\ldots,b,b,{\bf b},1,\ldots}, q_2}\\
&\redar&\tmach{{\mathcal S}, {\mathcal Q}, {\mathcal F}, \tape{\ldots,b,b,b,{\bf 1},\ldots}, q_3}\; .
\end{eqnarray*}
Since two is even this time the Turing Machine halts with the current head position on a $1$ symbol
as required.

\section{Intensional Process Calculi}
\label{sec:int}

Intensionality in process calculi is the ability for the structure of a term to be
determined, here during interaction.
This has appeared in communication primitives in some more recent process calculi such
as Concurrent Pattern Calculus and Psi Calculi.
Spi calculus supports intentional reductions, but not in a communication reduction
\cite{Abadi:1997:CCP:266420.266432}.
This section defines an {\em asymmetric concurrent pattern calculus} (ACPC) that
is a simple variation of Concurrent Pattern Calculus as described before \cite{GivenWilsonPhD}.
ACPC is trivial to represent in either Concurrent Pattern Calculus or Psi Calculi,
and so has been used here for clarity of the encoding, and to transfer the results
(details in Section~\ref{sec:var}).

Assume a countable collection of names denoted $m,m_1,m_2,n,n_1,\ldots$.
The {\em terms} of ACPC are given by
\begin{eqnarray*}
s,t &::=& n\ \bnf\ s\bullet t
\end{eqnarray*}
the {\em names} $n$, and {\em compounds} $s\bullet t$ that combine the two terms $s$ and $t$
into a single term.

The {\em patterns} of ACPC are defined by
\begin{eqnarray*}
p,q &::=& \l n\ \bnf\ n\ \bnf\ p\bullet q\; .
\end{eqnarray*}
The {\em binding names} $\l n$ play the r\^ole of inputs in the pattern.
The {\em name-match} $n$ is used to test for equality during interaction
\footnote{This corresponds to the protected names $\pro n$ of CPC in the
r\^ole they play. However, the syntax is chosen to mirror the variable names
$n$ of CPC since they more closely align with $\pi$-calculus and Psi Calculi syntax, and
later results for CPC can use either protected or variable names.}.
The {\em compound patterns} $p\bullet q$ combine the two patterns $p$ and $q$ into a single pattern.
Note that a well-formed pattern is one where each binding name appears only once, the
rest of this paper will only consider well formed patterns.

Substitutions, denoted $\sigma,\rho$, are finite mappings from names to terms.
Their domain and range are expected, with their names being the union of domain and range.

\renewcommand{\match}[2]{\{#1/\!\!/#2\}}

The key to interaction for ACPC is the matching $\match t p$ of the term $t$ against the pattern $p$
to generate a substitution $\sigma$
is defined as follows.
\begin{eqnarray*}
\match t {\l n} &\define& \{t/n\}\\
\match n n &\define& \{\}\\
\match {s\bullet t} {p\bullet q} &\define& \match s p \cup \match t q\\
\match t p &\mbox{undefined}& \mbox{otherwise.}
\end{eqnarray*}
Any term $t$ can be matched with a binding name $\l n$ to generate a substitution from the
binding name to the term $\{t/n\}$.
A single name $n$ can be matched with a name-match for that name $n$ to yield the
empty substitution.
A compound term $s\bullet t$ can be matched by a compound pattern $p\bullet q$ when
the components match to yield substitutions $\match s p=\sigma_1$ and $\match t q=\sigma_2$,
the resulting substitution is the unification of $\sigma_1$ and $\sigma_2$.
Observe that since patterns are well formed, the substitutions of components will always have
disjoint domain.
Otherwise the match is undefined.

The processes of ACPC are given by
\begin{eqnarray*}
P,Q &::=& \zero\BNF P\bnf Q\BNF !P\BNF \res n P\BNF \iap p P\BNF \oap t P\; .
\end{eqnarray*}
The null process, parallel composition, replication, and restriction are standard
from CPC (and many other process calculi).
The {\em input} $\iap p P$ has a pattern $p$ and body $P$, the binding names of the pattern
bind their instances in the body.
The {\em output} $\oap t P$ has a term $t$ and body $P$, like in $\pi$-calculus and Psi
Calculi there are no binding names or scope effects for outputs.
Note that an input $\iap p \zero$ may be denoted by $p$ and an output $\oap t \zero$
by $\oan t$ when no ambiguity may occur.

$\alpha$-conversion $=_\alpha$ is defined upon inputs and restrictions in the usual manner for 
Concurrent Pattern Calculus \cite{GivenWilsonGorlaJay10}.
The structural equivalence relation $\equiv$ is given by:
\begin{equation*}
\begin{array}{c}
P\bnf \zero \equiv P
\qquad\qquad\qquad
P\bnf Q \equiv Q \bnf P
\qquad\qquad\qquad
P\bnf (Q\bnf R) \equiv (P\bnf Q)\bnf R
\vspace{0.2cm} \\
P\equiv P'\quad\mbox{if}\ P=_\alpha P'
\qquad\qquad\qquad
\res a \zero \equiv \zero
\qquad\qquad\qquad
\res a \res b P\equiv \res b \res a P
\vspace{0.2cm} \\
P\bnf \res a Q\equiv \res a (P\bnf Q)\quad \mbox{if}\ a\notin{\sf fn}(P)
\qquad\qquad\qquad
!P\equiv P\bnf !P \; .
\end{array}
\end{equation*}

The application of a substitution $\sigma$ to a process $P$ denoted $\sigma P$ is
as usual with scope capture avoided by $\alpha$-conversion where required.

ACPC has a single interaction axiom given by:
\begin{eqnarray*}
\oap t P \BNF \iap q Q &\quad\redar\quad& P \BNF \sigma Q\qquad\qquad\qquad \match t q = \sigma\;.
\end{eqnarray*}
It states that when the term of an output can be matched with the pattern of an input to
yield the substitution $\sigma$ then reduce to the body of the output in parallel with
$\sigma$ applied to the body of the input.

\section{Encoding}
\label{sec:enc}

\newcommand{\enclist}[2]{#2} 

This section presents a faithful encoding of Turing Machines into ACPC.
The key to the encoding is to capture the current state and tape as the term of
an output, and to encode the transition function as a process that will operate
upon the encoded tape.
The spirit to this kind of encoding has been captured before when encoding
combinatory logics into CPC \cite{GivenWilsonPhD}.

Consider the simple encodings $\ytrans \cdot _L$ and $\ytrans \cdot _R$ that take a sequence of symbols
and encodes them into a term as follows:
\enclist{
\begin{eqnarray*}
\ytrans{\ldots,s_3,s_2,s_1} _L &\define& s_1\bullet (s_2\bullet (s_3\bullet \ldots))\\
\ytrans{s_1,s_2,s_3,\ldots} _R &\define& s_1\bullet (s_2\bullet (s_3\bullet \ldots))\; .
\end{eqnarray*}
That is, $\ytrans \cdot _L$ encodes a sequence of symbols (to be from the left hand side of
the current head position) into a list starting with the symbol on the right (closest to the current
head position) at the head of the list.
Similarly, $\ytrans \cdot _R$ encodes a sequence of symbols (from the right hand side of
the current head position) into a list starting from the symbol on the left (again, closest to the current
head position).
}
{
\begin{eqnarray*}
\ytrans{\ldots,s_3,s_2,s_1} _L &\define& ((\ldots\bullet s_3)\bullet s_2)\bullet s_1\\
\ytrans{s_1,s_2,s_3,\ldots} _R &\define& s_1\bullet (s_2\bullet (s_3\bullet \ldots))\; .
\end{eqnarray*}
That is, $\ytrans \cdot _L$ encodes a sequence of symbols from right to left, compounding
on the left hand side.
Similarly, $\ytrans \cdot _R$ encodes a sequence of symbols from left to right,
compounding on the right hand side.
}

Now consider a tape that must be of the form
$\tape{\ldots,b,s_a,\ldots,s_g,{\bf s_h},s_i\ldots,s_j,b,\ldots}$.
That is, an infinite sequence of blanks, some sequence of symbols including the current
head position, and then another infinite sequence of blanks.
This can be encoded $\xtrans\cdot$ into a term by:
\begin{eqnarray*}
\xtrans{\ \tape{\ldots,b,s_a,\ldots,s_g,{\bf s_h},s_i\ldots,s_j,b,\ldots}\ }
&\define&
(\ytrans{\ldots,b,s_a,\ldots,s_g}_L)\bullet s_h\bullet (\ytrans {s_i\ldots,s_j,b,\ldots}_R)
\end{eqnarray*}
\enclist{
Observe that the result is a term of the form $a\bullet b\bullet c$
where $a$ is the encoding of the tape left of the current head position
(reversed so as to be ordered from right to left),
$b$ is the current head position, and $c$ is the encoding of the tape right of
the current head position.
In particular note that both $a$ and $c$ are of the form $s_1\bullet ss$ where
$s_1$ is the symbol next to the current head position and $ss$ is the rest of
the sequence in that direction.
}
{
Observe that the result is a term of the form $a\bullet b\bullet c$
where $a$ is the encoding of the tape left of the current head position,
$b$ is the current head position symbol,
and $c$ is the encoding of the tape right of the current head position.
In particular note that $a$ and $c$ are both compounds of their symbol closest
to the current head position and the rest of the sequence in their direction.
}
For now the encoding handles an infinite tape and produces an infinite term,
although this can be removed later without effect on the results (details in Section~\ref{sec:var}).

Now the current state $q_i$ and a tape ${\mathcal T}$ can be represented as an output by
\begin{equation*}
\oan {q_i\bullet \encode {\mathcal T}}\; .
\end{equation*}

\begin{lemma}
\label{lem:tape-no-red}
The representation $\oan {q_i\bullet \encode {\mathcal T}}$ of the state $q_i$
and tape ${\mathcal T}$ does not reduce.
\end{lemma}

With the current state and tape encoded into a term it remains to encode the
transition function in a manner that allows faithfulness.
Consider that each tuple of the transition function is of the form
$\tuple {q_i,s_1,s_2,d,q_j}$ for states $q_i$ and $q_j$, and symbols $s_1$ and $s_2$, and
for $d$ either $L$ or $R$. Thus we can encode $\encode\cdot$ such a tuple as a single process
as follows:
\enclist{
\begin{equation*}
\begin{array}{rcll}
\encode{\tuple {q_i,s_1,s_2,d,q_j}}
&\define&
\iap {q_i\bullet ((\l l_1\bullet \l l)\bullet s_1\bullet \l r)}
\oan {q_j\bullet (l\bullet l_1\bullet (s_2\bullet r))}\qquad&d=L\\
\encode{\tuple {q_i,s_1,s_2,d,q_j}}
&\define&
\iap {q_i\bullet (\l l\bullet s_1\bullet (\l r_1\bullet \l r))}
\oan {q_j\bullet ((s_2\bullet l)\bullet r_1\bullet r)}\qquad&d=R\; .
\end{array}
\end{equation*}
}
{
\begin{equation*}
\begin{array}{rcll}
\encode{\tuple {q_i,s_1,s_2,d,q_j}}
&\define&
\iap {q_i\bullet ((\l l\bullet \l l_1)\bullet s_1\bullet \l r)}
\oan {q_j\bullet (l\bullet l_1\bullet (s_2\bullet r))}\qquad&d=L\\
\encode{\tuple {q_i,s_1,s_2,d,q_j}}
&\define&
\iap {q_i\bullet (\l l\bullet s_1\bullet (\l r_1\bullet \l r))}
\oan {q_j\bullet ((l\bullet s_2)\bullet r_1\bullet r)}\qquad&d=R\; .
\end{array}
\end{equation*}
}
Note that in both cases the pattern matches on the state $q_i$ and the symbol $s_1$.
When the tape is going to move left then the first symbol to the left is bound to $l_1$
and the rest to $l$, and $r_1$ and $r$ when respectively moving right.
The output in each case is the new state $q_j$ and the tape with the written symbol
$s_2$ added to the right side of the head position when moving left, or the left side
when moving right.
Thus, the new output represents the updated state and tape after the transition
$\tuple {q_i,s_1,s_2,d,q_j}$ has been applied once.
Note that the encoding here assumes the four names $l$ and $l_1$ and $r$ and $r_1$ do
not appear in the symbols ${\mathcal S}$ of the encoded Turing Machine.
Since the collection of symbols ${\mathcal S}$ is finite it is always possible to
choose four such names.

Building upon this, the encoding of the transition function ${\mathcal F}$ can be
done.
Define $\prod_{x\in S} P(x)$ to be the parallel composition of processes $P(x)$ in
the usual manner.
Now the encoding of the transition function $\encode{{\mathcal F}}$ can be captured as follows:
\begin{eqnarray*}
\encode{\mathcal F} &\define& 
\prod _{u\in{\mathcal F}} \ !\encode u
\end{eqnarray*}
where $u$ is each tuple of the form $\tuple {q_i,s_1,s_2,d,q_j}$.
Observe that this creates a process of the form $!P_1\bnf !P_2\bnf \ldots$ where each
$P_i$ performs a single transition.

\begin{lemma}
\label{lem:trans-no-red}
The encoding $\encode{{\mathcal F}}$ of the transition function ${\mathcal F}$
does not reduce.
\end{lemma}
\begin{proof}
For every tuple $u$ in ${\mathcal F}$ the encoding $\encode u$ is an input.
It is then straightforward to consider all the structural congruence rules
and show that there are no outputs and thus the reduction axiom cannot be satisfied.
\end{proof}

Finally, the encoding $\encode\cdot$ of a Turing Machine into ACPC is given by:
\begin{eqnarray*}
\encode{\tmach{{\mathcal S}, {\mathcal Q}, {\mathcal F}, {\mathcal T}, q_i}}
&\quad\define\quad&
\oan {q_i\bullet \encode {\mathcal T}}
\BNF
\encode{\mathcal F}\; .
\end{eqnarray*}

\begin{lemma}
\label{lem:faithful}
Given a Turing Machine $\tmach{{\mathcal S}, {\mathcal Q}, {\mathcal F}, {\mathcal T}, q_i}$
then
\begin{enumerate}
\item If there is a transition
$\tmach{{\mathcal S}, {\mathcal Q}, {\mathcal F}, {\mathcal T}, q_i}\redar
\tmach{{\mathcal S}, {\mathcal Q}, {\mathcal F}, {\mathcal T}', q_j}$
then there is a reduction
$\encode{\tmach{{\mathcal S}, {\mathcal Q}, {\mathcal F}, {\mathcal T}, q_i}}\redar Q$
where $Q\equiv\encode{\tmach{{\mathcal S}, {\mathcal Q}, {\mathcal F}, {\mathcal T}', q_j}}$, and
\item if there is a reduction
$\encode{\tmach{{\mathcal S}, {\mathcal Q}, {\mathcal F}, {\mathcal T}, q_i}}\redar Q$
then $Q\equiv\encode{\tmach{{\mathcal S}, {\mathcal Q}, {\mathcal F}, {\mathcal T}', q_j}}$ and
there is a transition
$\tmach{{\mathcal S}, {\mathcal Q}, {\mathcal F}, {\mathcal T}, q_i}\redar
\tmach{{\mathcal S}, {\mathcal Q}, {\mathcal F}, {\mathcal T}', q_j}$.
\end{enumerate}
\end{lemma}
\begin{proof}
The first part is proven by examining the tuple $u$ that corresponds to the transition 
$\tmach{{\mathcal S}, {\mathcal Q}, {\mathcal F}, {\mathcal T}, q_i}\redar
\tmach{{\mathcal S}, {\mathcal Q}, {\mathcal F}, {\mathcal T}', q_j}$
by the Turing Machine. This tuple must be of the form
$\tuple{q_i,s_1,s_2,d,q_j}$ and it must also be that ${\mathcal T}$ is of the form
$\tape{\ldots,s_m,{\bf s_1},s_n,\ldots}$.
Further, it must be that ${\mathcal T}'$ is either:
$\tape{\ldots,{\bf s_m},s_2,s_n,\ldots}$ when $d$ is $L$, or
$\tape{\ldots,s_m,s_2,{\bf s_n},\ldots}$ when $d$ is $R$.
Now $\encode{\tmach{{\mathcal S}, {\mathcal Q}, {\mathcal F}, {\mathcal T}, q_i}}$
is of the form
$\oan {q_i\bullet \encode {\mathcal T}} \BNF \encode{\mathcal F}$ and
by Lemmas~\ref{lem:tape-no-red} and \ref{lem:trans-no-red} neither $\oan {q_i\bullet \encode {\mathcal T}}$
nor $\encode{\mathcal F}$ can reduce, respectively.
Now by exploiting structural congruence gain that
$\oan {q_i\bullet \encode {\mathcal T}} \BNF \encode{\mathcal F}
\equiv
\oan {q_i\bullet \encode {\mathcal T}} \BNF \encode {\tuple{q_i,s_1,s_2,d,q_j}}
\BNF \encode{\mathcal F}$.
Then by the definition of matching and the reduction axiom is straightforward to show that
$\oan {q_i\bullet \encode {\mathcal T}} \BNF \encode {\tuple{q_i,s_1,s_2,d,q_j}}\redar
 \oan {q_j\bullet \encode {{\mathcal T}'}}$ and thus conclude.

The reverse direction is proved similarly by observing that the only possible reduction
$\encode{\tmach{{\mathcal S}, {\mathcal Q}, {\mathcal F}, {\mathcal T}, q_i}}\redar Q$
must be due to a tuple $\tuple{q_i,s_1,s_2,d,q_j}$ that is in the transition function ${\mathcal F}$
and the result follows.
\end{proof}

\begin{theorem}
\label{thm:done}
The encoding $\encode\cdot$ of a Turing Machine into ACPC; faithfully preserves reduction, and
divergence.
That is, given a Turing Machine $\tmach{{\mathcal S}, {\mathcal Q}, {\mathcal F}, {\mathcal T}, q_i}$
then it holds that:
\begin{enumerate}
\item there is a transition
$\tmach{{\mathcal S}, {\mathcal Q}, {\mathcal F}, {\mathcal T}, q_i}\redar
\tmach{{\mathcal S}, {\mathcal Q}, {\mathcal F}, {\mathcal T}', q_j}$
if and only if there is exactly one reduction
$\encode{\tmach{{\mathcal S}, {\mathcal Q}, {\mathcal F}, {\mathcal T}, q_i}}\redar Q$
where $Q\equiv\encode{\tmach{{\mathcal S}, {\mathcal Q}, {\mathcal F}, {\mathcal T}', q_j}}$, and
\item there is an infinite sequence of transitions
$\tmach{{\mathcal S}, {\mathcal Q}, {\mathcal F}, {\mathcal T}, q_i}\redar^\omega$
if and only if there is an infinite sequence of reductions
$\encode{\tmach{{\mathcal S}, {\mathcal Q}, {\mathcal F}, {\mathcal T}, q_i}}\redar^\omega$.
\end{enumerate}
\end{theorem}
\begin{proof}
Both parts can be proved by exploiting Lemma~\ref{lem:faithful}.
\end{proof}

Observe that this encoding of a Turing Machine into ACPC is not only faithful
and straightforward, but also up to structural congruence.
This is in contrast with the popular style of encoding $\l$-calculi into process
calculi that requires many reductions to simulate one $\l$-reduction, and the
equivalence of encoded terms/machines is only up to weak behavioural equivalence.
The simplicity and faithfulness here is gained by being able to directly render
the current state and tape as a single term, and the transition function as a
process that modifies the current state and tape in the same manner as the original
Turing Machine.

\section{Variations}
\label{sec:var}

This section considers variations to the encoding including:
representing the tape as a finite term,
encoding into Concurrent Pattern Calculus,
and encoding into Psi Calculi.

\subsection*{Finite Terms}

One potential concern is the infinite tape being represented as an infinite term
in ACPC. However, this can be done away with by adding an additional reserved name $e$
during the translation that does not appear in the symbols of the Turing Machine ${\mathcal S}$
and represents the edge of the tape.

Now $\ytrans \cdot _L$ and $\ytrans \cdot _R$ are modified to account for the
endless sequence of blank symbols $b$ as follows:
\enclist{
\begin{eqnarray*}
\ytrans{\ldots,b,s_i,\ldots,s_1} _L &\define& s_1\bullet(\ldots(s_i\bullet e))\\
\ytrans{s_1,\ldots,s_i,b,\ldots} _R &\define& s_1\bullet(\ldots(s_i\bullet e))\; .
\end{eqnarray*}}{
\begin{eqnarray*}
\ytrans{\ldots,b,s_i,\ldots,s_1} _L &\define& ((e\bullet s_i)\ldots)\bullet s_1\\
\ytrans{s_1,\ldots,s_i,b,\ldots} _R &\define& s_1\bullet(\ldots(s_i\bullet e))\; .
\end{eqnarray*}}
Here the endless blanks at the edge of the tape are simply replaced by $e$.
Otherwise the encoding of the state $q_i$ and tape ${\mathcal T}$ is the same.

\begin{lemma}
\label{lem:fin:tape-no-red}
The representation $\oan {q_i\bullet \encode {\mathcal T}}$ of the state $q_i$
and tape ${\mathcal T}$ does not reduce.
\end{lemma}
%

The encoding of tuples $\encode\cdot$ is now modified to account for $e$ given by:
\begin{equation*}
\begin{array}{rcll}
\encode{\tuple {q_i,s_1,s_2,d,q_j}}
&\define&\ \ \ 
!\iap {q_i\bullet ((\enclist{\l l_1\bullet \l l}{\l l\bullet \l l_1})\bullet s_1\bullet \l r)}
\oan {q_j\bullet (l\bullet l_1\bullet (s_2\bullet r))}\\ 
\ &\ &\bnf
!\iap {q_i\bullet (e\bullet s_1\bullet \l r)}
\oan {q_j\bullet (e\bullet b\bullet (s_2\bullet r))}\qquad&d=L\\
\encode{\tuple {q_i,s_1,s_2,d,q_j}}
&\define&\ \ \ 
!\iap {q_i\bullet (\l l\bullet s_1\bullet (\l r_1\bullet \l r))}
\oan {q_j\bullet ((\enclist{s_2\bullet l}{l\bullet s_2})\bullet r_1\bullet r)}\\ 
\ &\ &\bnf
!\iap {q_i\bullet (\l l\bullet s_1\bullet e)}
\oan {q_j\bullet ((\enclist{s_2\bullet l}{l\bullet s_2})\bullet b\bullet e)}\qquad&d=R\; .
\end{array}
\end{equation*}
The encoding of a tuple now has two input processes in parallel and each under
a replication;
the first matching the original encoding, and the second detecting when the
transition would move the current head position to the edge of the tape.
The new one inserts a new blank $b$ in the output and shifts the edge $e$ along
one cell.
Observe that due to definition of the matching rule no output can interact
with both of these inputs (as no term can be matched with both patterns of the
form $\l m\bullet \l n$ and $e$).
The replications have been added so that structural congruence can be achieved in
the final results. This requires a change to the encoding of the transitions
function as follows:
\begin{eqnarray*}
\encode{\mathcal F} &\define& 
\prod _{u\in{\mathcal F}} \ \encode u
\end{eqnarray*}
where the replications are now left to the encoding of each tuple $\encode u$.

The rest of the results follow with minor alterations.

\begin{lemma}
\label{lem:fin:trans-no-red}
The encoding $\encode{{\mathcal F}}$ of the transition function ${\mathcal F}$
does not reduce.
\end{lemma}

\begin{lemma}
\label{lem:fin:faithful}
Given a Turing Machine $\tmach{{\mathcal S}, {\mathcal Q}, {\mathcal F}, {\mathcal T}, q_i}$
then
\begin{enumerate}
\item If there is a transition
$\tmach{{\mathcal S}, {\mathcal Q}, {\mathcal F}, {\mathcal T}, q_i}\redar
\tmach{{\mathcal S}, {\mathcal Q}, {\mathcal F}, {\mathcal T}', q_j}$
then there is a reduction
$\encode{\tmach{{\mathcal S}, {\mathcal Q}, {\mathcal F}, {\mathcal T}, q_i}}\redar Q$
where $Q\equiv\encode{\tmach{{\mathcal S}, {\mathcal Q}, {\mathcal F}, {\mathcal T}', q_j}}$, and
\item if there is a reduction
$\encode{\tmach{{\mathcal S}, {\mathcal Q}, {\mathcal F}, {\mathcal T}, q_i}}\redar Q$
then $Q\equiv\encode{\tmach{{\mathcal S}, {\mathcal Q}, {\mathcal F}, {\mathcal T}', q_j}}$ and
there is a transition
$\tmach{{\mathcal S}, {\mathcal Q}, {\mathcal F}, {\mathcal T}, q_i}\redar
\tmach{{\mathcal S}, {\mathcal Q}, {\mathcal F}, {\mathcal T}', q_j}$.
\end{enumerate}
\end{lemma}
\begin{proof}
The first part is proven by examining the tuple $u$ that corresponds to the transition 
$\tmach{{\mathcal S}, {\mathcal Q}, {\mathcal F}, {\mathcal T}, q_i}\redar
\tmach{{\mathcal S}, {\mathcal Q}, {\mathcal F}, {\mathcal T}', q_j}$
by the Turing Machine. This tuple must be of the form
$\tuple{q_i,s_1,s_2,d,q_j}$ and it must also be that ${\mathcal T}$ is of the form
$\tape{\ldots,s_m,{\bf s_1},s_n,\ldots}$.
Further, it must be that ${\mathcal T}'$ is either:
$\tape{\ldots,{\bf s_m},s_2,s_n,\ldots}$ when $d$ is $L$, or
$\tape{\ldots,s_m,s_2,{\bf s_n},\ldots}$ when $d$ is $R$.
Now $\encode{\tmach{{\mathcal S}, {\mathcal Q}, {\mathcal F}, {\mathcal T}, q_i}}$
is of the form
$\oan {q_i\bullet \encode {\mathcal T}} \BNF \encode{{\mathcal F}}$ and
by Lemmas~\ref{lem:fin:tape-no-red} and \ref{lem:fin:trans-no-red} neither $\oan {q_i\bullet \encode {\mathcal T}}$
nor $\encode{{\mathcal F}}$ can reduce, respectively.

By definition $\encode{\mathcal F}$ is of the form
$\encode{u}\bnf R$
for some process $R$. 
Now consider $d$.
\begin{itemize}
\item If $d$ is $L$ then consider the encoded tape $\encode{{\mathcal T}}$.
  \begin{itemize}
  \item If $\encode{{\mathcal T}}$ is of the form $e\bullet s_1\bullet t$
        then by definition of $\encode{u}$ and structural congruence $\encode{u}\equiv
        \iap {q_i\bullet (e\bullet s_1\bullet \l r)}\oan {q_j\bullet (e\bullet b\bullet (s_2\bullet r))}
        \bnf \encode{u}$ and thus there is a reduction
        $\oan {q_i\bullet \encode {\mathcal T}} \bnf
        \iap {q_i\bullet (e\bullet s_1\bullet \l r)}\oan {q_j\bullet (e\bullet b\bullet (s_2\bullet r))}
        \bnf \encode{u}\redar \oan {q_j\bullet(e\bullet b\bullet (s_2\bullet r))}\bnf \encode u$.
        It is straightforward to show that $\encode{{\mathcal T}'}=e\bullet b\bullet (s_2\bullet r)$
        and thus by structural congruence that
        $\oan {q_j\bullet(e\bullet b\bullet (s_2\bullet r))}\bnf \encode u\bnf R\equiv
        \oan {q_j\bullet \encode{{\mathcal T}'}}\bnf \encode{{\mathcal F}}$ and thus conclude.
  \item If $\encode{{\mathcal T}}$ is of the form
        $(\enclist{s_i\bullet s}{s\bullet s_i})\bullet s_1\bullet r$
        then take $\encode u\equiv
        \iap {q_i\bullet ((\enclist{\l l_1\bullet \l l}{\l l\bullet\l l_1})\bullet s_1\bullet \l r)}
        \oan {q_j\bullet (l\bullet l_1\bullet (s_2\bullet r))}
        \bnf \encode{u}$ and the rest is as in the previous case.
  \end{itemize}
\item If $d$ is $R$ then the proof is a straightforward adaptation of the $L$ case above.
\end{itemize}

The reverse direction is proved similarly by observing that the only possible reduction
$\encode{\tmach{{\mathcal S}, {\mathcal Q}, {\mathcal F}, {\mathcal T}, q_i}}\redar Q$
must be due to a tuple $\tuple{q_i,s_1,s_2,d,q_j}$ that is in the transition function ${\mathcal F}$
and the result follows. The only added complexity is to ensure that there is only one possible
reduction for a given current state and current head position symbol, this can be assured by
definition of the match rule excluding any term from matching with both patterns
$\enclist{\l l_1\bullet \l l}{\l l\bullet \l l_1}$ and $e$.
\end{proof}

\begin{theorem}
\label{thm:fin:done}
The encoding $\encode\cdot$ of a Turing Machine into ACPC; faithfully preserves reduction, and
divergence.
That is, given a Turing Machine $\tmach{{\mathcal S}, {\mathcal Q}, {\mathcal F}, {\mathcal T}, q_i}$
then it holds that:
\begin{enumerate}
\item there is a transition
$\tmach{{\mathcal S}, {\mathcal Q}, {\mathcal F}, {\mathcal T}, q_i}\redar
\tmach{{\mathcal S}, {\mathcal Q}, {\mathcal F}, {\mathcal T}', q_j}$
if and only if there is exactly one reduction
$\encode{\tmach{{\mathcal S}, {\mathcal Q}, {\mathcal F}, {\mathcal T}, q_i}}\redar Q$
where $Q\equiv\encode{\tmach{{\mathcal S}, {\mathcal Q}, {\mathcal F}, {\mathcal T}', q_j}}$, and
\item there is an infinite sequence of transitions
$\tmach{{\mathcal S}, {\mathcal Q}, {\mathcal F}, {\mathcal T}, q_i}\redar^\omega$
if and only if there is an infinite sequence of reductions
$\encode{\tmach{{\mathcal S}, {\mathcal Q}, {\mathcal F}, {\mathcal T}, q_i}}\redar^\omega$.
\end{enumerate}
\end{theorem}
\begin{proof}
Both parts can be proved by exploiting Lemma~\ref{lem:fin:faithful}.
\end{proof}

\subsubsection*{Concurrent Pattern Calculus}
\label{ssec:cpc}


The choice of using ACPC here rather than CPC is for simplicity in presentation.
This section recalls CPC and proves that the encodings hold in CPC as well.
CPC has a single class of {\em patterns} that combines both the terms and patterns of
ACPC defined as follows:
\begin{eqnarray*}
p,q &::=& \lambda x \BNF n \BNF \pro n\BNF p\bullet q\; .
\end{eqnarray*}
The {\em binding names} $\lambda x$ are as before.
The {\em variable names} $n$ can be used as both output (like the name terms
of ACPC) and equality tests (like the name-match of ACPC).
The {\em protected names} $\pro n$ are only equality tests (name-matches of ACPC).
{\em Compounds} $p\bullet q$ are as in ACPC.
A {\em communicable pattern} is a pattern that contains no binding or protected names.
\newcommand{\cpcmatch}[2]{\{ #1\pmatch #2\}}

Interaction CPC relies upon the {\em unification} $\cpcmatch p q$ of the patterns $p$ and $q$
to yield a pair of substitutions $(\sigma,\rho)$ and is 
defined by:
\begin{equation*}
\begin{array}{rcll}
\left.
\begin{array}{r}
\cpcmatch x x \\
\cpcmatch x {\pro x} \\
\cpcmatch {\pro x} x \\
\cpcmatch {\pro x} {\pro x}
\end{array}
\right\}
&\define& (\{\},\{\})\\
\cpcmatch {\lambda x} q &\define& (\{q/x\},\{\}) &q\mbox{\ is communicable}\\
\cpcmatch p {\lambda x} &\define& (\{\},\{p/x\}) &p\mbox{\ is communicable}\\
\cpcmatch {p_1\bullet p_2} {q_1\bullet q_2} &\define& (\sigma_1\cup\sigma_2,\rho_1\cup\rho_2)
& \cpcmatch {p_i} {q_i} = (\sigma_i,\rho_i)\ i\in\{1,2\}\\
\cpcmatch p q &\mbox{undefined}& \mbox{otherwise.}
\end{array}
\end{equation*}
The unification succeeds and yields empty substitutions when both patterns are the same
name and are both variable or protected.
If either pattern is a binding name and the other is communicable, then the communicable
pattern is bound to the binding name in the appropriate substitution.
Otherwise if both patterns are compounds then unify component-wise.
Finally, if all these fail then unification is undefined (impossible).

The process of CPC are given by:
\begin{eqnarray*}
P,Q &::=& \zero\BNF P|Q\BNF !P\BNF \res n P\BNF \iap p P\; .
\end{eqnarray*}
All are familiar from ACPC although the input and output are now both represented by
the {\em case} $\iap p P$ with pattern $p$ and body $P$.

The structural laws are the same as for ACPC with $\alpha$-conversion defined in the usual
manner \cite{GivenWilsonGorlaJay10,givenwilson:hal-00987578} and interaction is defined by the following
axiom:
\begin{eqnarray*}
\iap p P\bnf \iap q Q &\redar& (\sigma P)\bnf (\rho Q) \quad\quad\quad \cpcmatch p q = (\sigma,\rho)\; .
\end{eqnarray*}
It states that when two cases in parallel can unify their patterns to yield substitutions
$\sigma$ and $\rho$ then apply those substitutions to the appropriate bodies.


The encodings of Turing Machines into CPC are trivial, the only change is to remove
the overhead line from outputs, i.e.~$\oap t P$ becomes $\iap t P$ since all terms of
ACPC are patterns of CPC
\footnote{There is no need to convert syntax between ACPC and CPC, for example changing
name-matches from $n$ to $\pro n$ in patterns, as the unification rules for CPC allow for both.
Indeed, the encodings were chosen to allow this.
Although in theory
CPC could allow two ACPC outputs to interact, this does not occur for the encodings in this paper.}.
However some proofs require changes due to the change from input and
output with one sided matching, to CPC cases with pattern unification.
The proof of Lemma~\ref{lem:tape-no-red} is trivial. For Lemmas~\ref{lem:trans-no-red}
and \ref{lem:fin:trans-no-red}
the proof is resolved due to CPC unification only allowing a binding name $\l x$ to
unify with a communicable pattern.
The rest are effectively unchanged.

\begin{theorem}
\label{thm:cpc:done}
There is an encoding $\encode\cdot$ of a Turing Machine into CPC that;
faithfully preserves reduction, and
divergence.
That is, given a Turing Machine $\tmach{{\mathcal S}, {\mathcal Q}, {\mathcal F}, {\mathcal T}, q_i}$
then it holds that:
\begin{enumerate}
\item there is a transition
$\tmach{{\mathcal S}, {\mathcal Q}, {\mathcal F}, {\mathcal T}, q_i}\redar
\tmach{{\mathcal S}, {\mathcal Q}, {\mathcal F}, {\mathcal T}', q_j}$
if and only if there is exactly one reduction
$\encode{\tmach{{\mathcal S}, {\mathcal Q}, {\mathcal F}, {\mathcal T}, q_i}}\redar Q$
where $Q\equiv\encode{\tmach{{\mathcal S}, {\mathcal Q}, {\mathcal F}, {\mathcal T}', q_j}}$, and
\item there is an infinite sequence of transitions
$\tmach{{\mathcal S}, {\mathcal Q}, {\mathcal F}, {\mathcal T}, q_i}\redar^\omega$
if and only if there is an infinite sequence of reductions
$\encode{\tmach{{\mathcal S}, {\mathcal Q}, {\mathcal F}, {\mathcal T}, q_i}}\redar^\omega$.
\end{enumerate}
\end{theorem}

\medskip


\newcommand{\cheq}{\stackrel\cdot\leftrightarrow}
\newcommand{\terms}{{\bf T}}
\newcommand{\assertion}{{\bf A}}
\newcommand{\one}{{\bf 1}}
\newcommand{\compose}{\otimes}
\newcommand{\assert}[1]{\llparenthesis \, #1 \, \rrparenthesis}
\newcommand{\fram}[1]{{\cal F}(#1)}

\subsubsection*{Psi Calculi}

\newcommand{\psiap}[2]{\underline{#1}(#2)}
\newcommand{\psoap}[2]{\overline {#1}(#2)}

Similarly both encodings can be easily adapted for Psi Calculi \cite{BJPV11}.
Psi Calculi are parametrized with respect to two sets: terms $\terms$ defined by
\begin{eqnarray*}
M,N &::=& m\BNF M,N
\end{eqnarray*}
the names $m$ and the {\em pair} $M,N$ of two terms $M$ and $N$;
and assertions $\assertion$, ranged over by $\Psi$ (that play no significant r\^ole here).
The empty assertion is written $\one$.
Also assume two operators: channel equivalence, $\cheq \subseteq \terms \times \terms$,
and assertion composition, $\compose: \assertion \times \assertion \rightarrow \assertion$. 
It is also required that $\cheq$ is transitive and
symmetric, and that $(\compose,\one)$ is a commutative monoid.

Processes in Psi Calculi are defined as:
$$
P,Q\ ::=\ \zero\ \bnf\ 
 P|Q\ \bnf \ (\nu x)P\ \bnf \ !P\ \bnf \ 
 \psiap M {\lambda\wt x} N.P \ \bnf\ \psoap M N.P\ \bnf \ \assert\Psi
$$
exploiting the notation $\wt a$ for a sequence $a_1,\ldots,a_i$.
Most process forms are as usual with:
{\em input} $\psiap M {\lambda \wt x} N.P$ on channel $M$ and binding names $\wt x$ in the pattern $N$
and with body $P$;
and {\em output} $\psoap M N.P$ on channel $M$ and outputting term $N$.

The reduction relation semantics are given by
isolating the $\tau$ actions of the LTS given in \cite{BJPV11}.
To this aim, the definiton of frame of a process $P$, written $\fram P$,
as the set of unguarded assertions occurring in $P$. Formally:
$$
\fram{\assert \Psi} = \Psi
\qquad
\fram{(\nu x)P} = (\nu x)\fram P
\qquad
\fram{P|Q} = \fram P \compose \fram Q
$$
and is $\one$ in all other cases. Denote as $(\nu \wt b_P)\Psi_P$ the frame of $P$.
The structural laws are the same as in ACPC.
The reduction relation is inferred by the following axioms:
$$
\begin{array}{c}
\prooftree \Psi \vdash M \cheq N 
\justifies \Psi \triangleright \psoap M K.P\ |\ \psiap N {\lambda \wt x}H.Q \redar P\ |\ \{{\wt L}/{\wt x}\}Q 
\endprooftree\ K = H[\wt x := \wt L]

\qquad\qquad

\prooftree \Psi \triangleright P \redar P'
\justifies \Psi \triangleright (\nu x) P \redar (\nu x) P'
\endprooftree\ x \not\in \mbox{names}(\Psi)

\vspace*{.4cm}
\\

\prooftree \Psi \compose \Psi_Q \triangleright P \redar P'
\justifies \Psi \triangleright P\ |\ Q \redar P' \ |\ Q
\endprooftree\ \fram Q = (\nu \wt b_Q)\Psi_Q, \wt b_Q \mbox{ fresh for } \Psi \mbox{ and } P

\qquad\qquad

\prooftree P \equiv Q \quad \Psi \triangleright Q \redar Q' \quad Q' \equiv P'
\justifies \Psi \triangleright P \redar P'
\endprooftree

\end{array}
$$
The interesting axiom is the first that states when $M$ and $N$ are equivalent and
the term $K$ is equal to the term $H$ with each name in $x_i\in \wt x$ replaced by some $L_i$,
then reduce to $P$ in parallel with the substitution $\{L_i/x_i\}$ applied to $Q$.
Denote with $P \redar P'$ whenever $\one \triangleright P \redar P'$.


Both encodings of Turing Machines into ACPC can be easily adapted for Psi Calculi,
the following changes show how to do adapt the encodings for the infinite tape encoding.
All instances of the compounding operator $\bullet$ are replaced by the pair
operator ``$,$'', i.e.~all terms and patterns of the form $x\bullet y$ take the form $x,y$.
The encoding of the current state $q_i$ and tape ${\mathcal T}$ becomes:
\begin{equation*}
\psoap {q_i} {\ \encode {\mathcal T}\ }.\zero\; .
\end{equation*}
The encoding of the tuples becomes:
\begin{equation*}
\begin{array}{rcll}
\encode{\tuple {q_i,s_1,s_2,d,q_j}}
&\define&
\underline {q_i} (\l l_1,\l l, \l r)((\enclist{l_1,l}{l,l_1}), s_1, r) \ .\ 
\overline {q_j} (l, l_1, (s_2, r)).\zero\qquad&d=L\\
\encode{\tuple {q_i,s_1,s_2,d,q_j}}
&\define&
\underline {q_i} (\l l,\l r_1,\l r)(l, s_1, (r_1, r)) \ . \ 
\overline {q_j} ((s_2, l), r_1, r).\zero\qquad&d=R\; .
\end{array}
\end{equation*}
From there the rest of the encoding remains the same and the results are
straightforward.

\begin{theorem}
\label{thm:psi:done}
The encoding $\encode\cdot$ of a Turing Machine into Psi Calculi;
faithfully preserves reduction, and
divergence.
That is, given a Turing Machine $\tmach{{\mathcal S}, {\mathcal Q}, {\mathcal F}, {\mathcal T}, q_i}$
then it holds that:
\begin{enumerate}
\item there is a transition
$\tmach{{\mathcal S}, {\mathcal Q}, {\mathcal F}, {\mathcal T}, q_i}\redar
\tmach{{\mathcal S}, {\mathcal Q}, {\mathcal F}, {\mathcal T}', q_j}$
if and only if there is exactly one reduction
$\encode{\tmach{{\mathcal S}, {\mathcal Q}, {\mathcal F}, {\mathcal T}, q_i}}\redar Q$
where $Q\equiv\encode{\tmach{{\mathcal S}, {\mathcal Q}, {\mathcal F}, {\mathcal T}', q_j}}$, and
\item there is an infinite sequence of transitions
$\tmach{{\mathcal S}, {\mathcal Q}, {\mathcal F}, {\mathcal T}, q_i}\redar^\omega$
if and only if there is an infinite sequence of reductions
$\encode{\tmach{{\mathcal S}, {\mathcal Q}, {\mathcal F}, {\mathcal T}, q_i}}\redar^\omega$.
\end{enumerate}
\end{theorem}

\section{Limitations}
\label{sec:fail}

\newcommand{\ifte}[4]{{\sf if}\ #1=#2\ {\sf then}\ #3\ {\sf else}\ #4}
\newcommand{\piiap}[2]{#1(#2)}
\newcommand{\pioap}[2]{\overline{#1}\langle #2\rangle}

This section discusses the difficulties of attempting to faithfully encode Turing Machines
into non-intensional calculi, particularly $\pi$-calculi. (Here for synchronous polyadic $\pi$-calculus, but
adaptations for asynchronous or monadic variations are straightforward, although may require
more reductions.)

The $\pi$-calculus processes are given by the following grammar:
$$
P \ ::= \ \zero \ \bnf \ 
P \bnf Q \ \bnf \ !P \ \bnf \ (\nu n) P \ \bnf \ 
\piiap a {\wt x}.P\ \bnf \ \pioap a {\wt b}.P\; .
$$
The null process, parallel composition, replication, and restriction are as usual.
The input $\piiap a {\wt x}.P$ has channel name $a$ and a sequence of binding names
$x_1,x_2,\ldots,x_i$ denoted by $\wt x$ and body $P$.
The output $\pioap a {\wt b} .P$ has channel name $a$ and sequence of output names
$b_1,b_2,\ldots,b_i$ denoted $\wt b$ and body $P$.
The length of a sequence $\wt x$ is denoted $|\wt x|$,
i.e.~$|x_1,x_2,\ldots,x_k|=k$.
$\alpha$-conversion and structural equivalence are as usual.

The only reduction axiom is
$$
\pioap m {\wt b}.P\bnf \piiap n {\wt x}.Q \quad\redar\quad P\bnf \{b_i/x_i\}Q
\qquad\qquad m=n \mbox{\ and\ } |\wt b|=|\wt x|\; .
$$
That is an output and input reduce if they have the same channel name and the
length of their output names and binding names are the same, reducing to
the body of the output, in parallel with
the substitution that binds each output name $b_i$ to the corresponding binding name $x_i$
applied to the body of the input.
The reduction relation is obtained by closing this reduction rule by parallel, restriction and the
same structural congruence relation defined for ACPC.


The first limitation of $\pi$-calculi is in the number of names that can be determined
equal in an interaction. In the Psi Calculi encoding 
the channel name is used to detect the state, this can also be used for $\pi$-calculi as
well.
However, the detection of the symbol at the current head position would require an additional
reduction.
There are two approaches that can resolve this first limitation while maintaining faithfulness.
The first solution is to account for both names by representing every possible
pair of state and symbol by a new name. That is, the encoding of a transition tuple
can be represented by
\begin{eqnarray*}
\encode{\tuple{q_1,s_2,s_3,d_4,q_5}} &\define& {q_1 s_2} (\ldots).P
\end{eqnarray*}
for some form of input $(\ldots)$ and process $P$. Here the state $q_1$ and current
head symbol $s_2$ are combined into a single name $q_1 s_2$ by the encoding.
The second approach is to use a structural equivalence rule such as
$\ifte m n P Q$ with the following rules
\begin{equation*}
\ifte m m P Q \equiv P \qquad\qquad \ifte m n P Q \equiv Q\quad m\neq n \; .
\end{equation*}
Now the encoding of all of the tuples of the form $\tuple{q_1,s_a,s_b,d,q_c}$
are of the form
\begin{eqnarray*}
\encode{\tuple{q_1,s_a,s_b,d,q_c}} &\define& q_1(x,\ldots).
\ifte x {s_1} {P1} { }
\bnf \hdots
\bnf \ifte x {s_i} {P_i} \zero
\end{eqnarray*}
for each possible symbol $s_1,s_2,\ldots,s_i\in{\mathcal S}$.
Here $P_j$ represents the process that does the transition for
$\tuple{q_1,s_j,s_x,d,q_y}$, that is the reductions that correspond to the
transition for the matching current head position symbol $s_j$, and are $\zero$
otherwise.

Note that there are other solutions to the problem of matching the state, such as
doing further reductions after binding the symbol at the current head position,
however these would immediately fail faithfulness.

The impossibility of encoding Turing Machines faithfully without intensionality
arises from the encoding of the tape into a single structure.
Since $\pi$-calculi cannot bind a structured term to a single name, it is impossible
to represent the infinite tape by a finite structure.
The closest is to take the traditional approach of using some name(s) to identify
where the rest of the structure (tape) can be obtained from.
For example, consider the followings of an encoding into $\pi$-calculus:
\begin{eqnarray*}
\ytrans {\ldots,b,s_i,\ldots,s_1} _L &\define&
  \res {x_1} \overline{l} \langle s_1,x_1\rangle \bnf
  \res {x_2} \overline{x_1} \langle s_2,x_2\rangle \bnf
  \ldots
  \res {x_i} \overline{x_{i-1}} \langle s_i,x_i\rangle \bnf
  \res {x_{i+1}} \overline{x_i} \langle b,x_{i+1}\rangle \bnf \ldots\\
\ytrans {s_1,\ldots,s_i,b,\ldots} _R &\define&
  \res {x_1} \overline{r} \langle s_1,x_1\rangle \bnf
  \res {x_2} \overline{x_1} \langle s_2,x_2\rangle \bnf
  \ldots
  \res {x_i} \overline{x_{i-1}} \langle s_i,x_i\rangle \bnf
  \res {x_{i+1}} \overline{x_i} \langle b,x_{i+1}\rangle \bnf \ldots\; .
\end{eqnarray*}
where $l$ and $r$ are reserved names for the left and right hand sides of the tape, respectively,
in the encoding.
Observe that in each case, the name can be used as a channel to input the symbol to
the left $l$ or right $r$ and the next name to use for the next symbol in that direction.
Note that parallel composition is used as this allows results to exploit structural
equivalence.

Using this approach the state $q_i$ and tape
$\tape{\ldots,b,s_a,\ldots,s_g,{\bf s_h},s_i\ldots,s_j,b,\ldots}$ can be encoded by
\begin{eqnarray*}
\xtrans{\ \tape{\ldots,b,s_a,\ldots,s_g,{\bf s_h},s_i\ldots,s_j,b,\ldots}\ }_{q_i}
&\define& q_i s_h\langle l,r\rangle.\zero \bnf
\ytrans{\ldots,b,s_a,\ldots,s_g}_L\bnf
\ytrans {s_i\ldots,s_j,b,\ldots}_R\; .
\end{eqnarray*}

Now a transition $\tuple{q_i, s_1, s_2, d, q_j}$ can be encoded as follows
(showing the $d=L$ case only):
\begin{eqnarray*}
\encode{\tuple{q_i, s_1, s_2, L, q_j}} &\define&
 q_i s_1 ( l_c, r_c ) .\\
& & l_c ( s_c, l_1 ) .\\
& &  (v r_1)\\
& & (\ifte {s_c} {s_{00}} {q_j s_{00}\langle l_1, r_1\rangle} {}\hdots
\ifte {s_c} {s_{kk}} {q_j s_{kk}\langle l_1, r_1\rangle} \zero\\
& & \bnf r_1 < s_2, r_c >)
\end{eqnarray*}
where each line after the encoding does as follows.
The $q_i s_1 ( l_c, r_c ) .$ matches the (encoded) state $q_i$ and current head position symbol $s_1$,
and binds the names
to access the left and right parts of the tape to $l_c$ and $r_c$, respectively.
Since the transition moves left, the $l_c ( s_c, l_1 ) .$ then reads the next symbol to the
left $s_c$ and the name to access the rest of the left hand side of the tape $l_1$.
A new name for the new right hand side of the tape is created with $(v r_1)$.
The next
line detects
the new symbol $s_c$ under the current head position by
comparing to each possible symbol $s_{00},s_{01},\ldots,s_{kk}\in{\mathcal S}$ and
outputting the new left and right hand tape access channel names $l_1$ and $r_1$ respectively
on the appropriately encoded channel name $q_j s_c$.
Finally, in parallel $r_1 \langle s_2, r_c \rangle$ provides the new right hand side of the tape.
The encoding of the $d=R$ transitions can be done similarly.

Putting all of these pieces together as in Section~\ref{sec:enc} allows similar results
to ACPC to be applied to $\pi$-calculi.

\begin{lemma}
\label{lem:tape-no-red-pi}
The representation
$\xtrans{\ \tape{\ldots,b,s_a,\ldots,s_g,{\bf s_h},s_i\ldots,s_j,b,\ldots}\ }_{q_i}
\define q_i s_h\langle l,r\rangle.\zero \bnf
\ytrans{\ldots,b,s_a,\ldots,s_g}_L\bnf
\ytrans {s_i\ldots,s_j,b,\ldots}_R$.
of the state $q_i$ and tape $\tape{\ldots,b,s_a,\ldots,s_g,{\bf s_h},s_i\ldots,s_j,b,\ldots}$
does not reduce.
\end{lemma}

\begin{lemma}
\label{lem:trans-no-red-pi}
The encoding $\encode{\mathcal F} \define \prod _{u\in{\mathcal F}} \ !\encode u$
where $u$ is each tuple of the form $\tuple {q_i,s_1,s_2,d,q_j}$
of the transition function ${\mathcal F}$
does not reduce.
\end{lemma}

Finally, the encoding $\encode\cdot$ of a Turing Machine into 
$\pi$-calculus is given by:
\begin{eqnarray*}
\encode{\tmach{{\mathcal S}, {\mathcal Q}, {\mathcal F}, {\mathcal T}, q_i}}
&\quad\define\quad&
\res l\res r (\ \encode {\mathcal T}_{q_i}
\BNF
\encode{\mathcal F}\ )\; .
\end{eqnarray*}

The limitations of $\pi$-calculi appear in the following lemma where the correspondence
of a single reduction in the original to a single reduction in the translation is lost,
instead a single reduction becomes two reductions (or more for some other $\pi$-calculi).
Further, the proof is complicated by now having to consider $\alpha$-equivalence of all
the restricted names.
The impact is also in the final theorem in this section where faithfulness is lost.

\begin{lemma}
\label{lem:reduction-pi}
Given a Turing Machine $\tmach{{\mathcal S}, {\mathcal Q}, {\mathcal F}, {\mathcal T}, q_i}$
then
\begin{enumerate}
\item If there is a transition
$\tmach{{\mathcal S}, {\mathcal Q}, {\mathcal F}, {\mathcal T}, q_i}\redar
\tmach{{\mathcal S}, {\mathcal Q}, {\mathcal F}, {\mathcal T}', q_j}$
then there are reductions
$\encode{\tmach{{\mathcal S}, {\mathcal Q}, {\mathcal F}, {\mathcal T}, q_i}}\redar\redar Q$
where $Q\equiv\encode{\tmach{{\mathcal S}, {\mathcal Q}, {\mathcal F}, {\mathcal T}', q_j}}$, and
\item if there is a reduction
$\encode{\tmach{{\mathcal S}, {\mathcal Q}, {\mathcal F}, {\mathcal T}, q_i}}\redar Q$
then there exists $Q'$ such that $Q\redar Q'$ and
$Q'\equiv\encode{\tmach{{\mathcal S}, {\mathcal Q}, {\mathcal F}, {\mathcal T}', q_j}}$ and
there is a transition
$\tmach{{\mathcal S}, {\mathcal Q}, {\mathcal F}, {\mathcal T}, q_i}\redar
\tmach{{\mathcal S}, {\mathcal Q}, {\mathcal F}, {\mathcal T}', q_j}$.
\end{enumerate}
\end{lemma}
\begin{proof}
The first part is proven by examining the tuple $u$ that corresponds to the transition 
$\tmach{{\mathcal S}, {\mathcal Q}, {\mathcal F}, {\mathcal T}, q_i}\redar
\tmach{{\mathcal S}, {\mathcal Q}, {\mathcal F}, {\mathcal T}', q_j}$
by the Turing Machine. This tuple must be of the form
$\tuple{q_i,s_1,s_2,d,q_j}$ and it must also be that ${\mathcal T}$ is of the form
$\tape{\ldots,s_k,s_m,{\bf s_1},s_n,s_o,\ldots}$.
Further, it must be that ${\mathcal T}'$ is either:
$\tape{\ldots,s_k,{\bf s_m},s_2,s_n,s_o,\ldots}$ when $d$ is $L$, or
$\tape{\ldots,s_k,s_m,s_2,{\bf s_n},s_o,\ldots}$ when $d$ is $R$.
Now $\encode{\tmach{{\mathcal S}, {\mathcal Q}, {\mathcal F}, {\mathcal T}, q_i}}$
is of the form
$\res l\res r(\ \encode {\mathcal T}_{q_i}\BNF \encode{\mathcal F}\ )$ and
by Lemmas~\ref{lem:tape-no-red-pi} and \ref{lem:trans-no-red-pi} neither $\encode {\mathcal T}_{q_i}$
nor $\encode{\mathcal F}$ can reduce, respectively.
Now by exploiting structural congruence gain that
$\res l\res r(\ \encode {\mathcal T}_{q_i} \BNF \encode{\mathcal F}\ )
\equiv
\res l\res r(\ \encode {\mathcal T}_{q_i} \BNF \encode {\tuple{q_i,s_1,s_2,d,q_j}}
\BNF \encode{\mathcal F}\ )$.
Then by two applications of the definition of the $\pi$-calculus reduction axiom 
it is straightforward to show that
$\encode {\mathcal T}_{q_i} \BNF \encode {\tuple{q_i,s_1,s_2,d,q_j}}\redar\redar
P$ for some $P$.
Now by two applications of $\alpha$-conversion it can be shown that
$P\equiv \res l \res r (\ q_j s_x\langle l, r\rangle.\zero \bnf P_L\bnf P_R\bnf \encode{\mathcal F} \ )$
for some $s_x$ and $P_L$ and $P_R$.
Now consider $d$.
\begin{itemize}
\item When $d=L$ it can be shown that $s_x=s_m$ and by two applications of induction on the
  indices of the restrictions $\res {x_i}$ of the left and right hand sides of the tape it can be shown
  that $P_L\equiv \ytrans {\ldots,s_k} _L$ and that $P_R\equiv \ytrans {s_2,s_n,s_o,\ldots} _R$.
\item When $d=R$ is can be shown that $s_x=s_n$ and by two applications of induction on the
  indices of the restrictions $\res {x_i}$ of the left and right hand sides of the tape it can be shown
  that $P_L\equiv \ytrans {\ldots,s_k,s_m,s_2} _L$ and that $P_R\equiv \ytrans {s_o,\ldots} _R$.
\end{itemize}

The reverse direction is proved similarly by observing that the only possible reduction
$\encode{\tmach{{\mathcal S}, {\mathcal Q}, {\mathcal F}, {\mathcal T}, q_i}}\redar Q$
must be due to a tuple $\tuple{q_i,s_1,s_2,d,q_j}$ that is in the transition function ${\mathcal F}$.
Then it follows that there exists a reduction $Q\redar Q'$ by definition of
$\encode {\tuple{q_i,s_1,s_2,d,q_j}}$.
Finally showing that $Q'\equiv\encode{\tmach{{\mathcal S}, {\mathcal Q}, {\mathcal F}, {\mathcal T}', q_j}}$
again requires tedious renaming of both sides of the encoded tape.
\end{proof}

\begin{theorem}
\label{thm:done-pi}
The encoding $\encode\cdot$ of a Turing Machine into (synchronous polyadic) $\pi$-calculus;
preserves reduction, and divergence.
That is, given a Turing Machine $\tmach{{\mathcal S}, {\mathcal Q}, {\mathcal F}, {\mathcal T}, q_i}$
then it holds that:
\begin{enumerate}
\item there is a transition
$\tmach{{\mathcal S}, {\mathcal Q}, {\mathcal F}, {\mathcal T}, q_i}\redar
\tmach{{\mathcal S}, {\mathcal Q}, {\mathcal F}, {\mathcal T}', q_j}$
if and only if there are reductions
$\encode{\tmach{{\mathcal S}, {\mathcal Q}, {\mathcal F}, {\mathcal T}, q_i}}\redar\redar Q$
where $Q\equiv\encode{\tmach{{\mathcal S}, {\mathcal Q}, {\mathcal F}, {\mathcal T}', q_j}}$, and
\item there is an infinite sequence of transitions
$\tmach{{\mathcal S}, {\mathcal Q}, {\mathcal F}, {\mathcal T}, q_i}\redar^\omega$
if and only if there is an infinite sequence of reductions
$\encode{\tmach{{\mathcal S}, {\mathcal Q}, {\mathcal F}, {\mathcal T}, q_i}}\redar^\omega$.
\end{enumerate}
\end{theorem}

\section{Conclusions}
\label{sec:conc}

The traditional approaches to encoding Turing Machines into process calculi
tend to be indirect and lead to complex and unclear results.
This is particularly true when the traditional path for process calculi is
taken by encoding a Turing Machine into $\l$-calculus and then into a 
process calculus.

Recent calculi with intensional communication allow the representation of
the current state and tape of a Turing Machine to be made clear and simple.
Similarly, the capture of each transition of the Turing Machine by an
input that transforms the state into a new output is a straightforward and
elegant solution.
The result is an encoding that is not only clearer and more direct, but also
faithful and holds up to structural equivalence.

The encoding can also be adapted in various ways.
The infinite terms of an infinite tape can be made finite if the tape of
the Turing Machine has some finite sequence of symbols with infinite blanks
on either side.
The choice of asymmetric concurrent pattern calculus here is for clarity alone,
the results also hold with only minor adaptations for both Concurrent Pattern
Calculus and Psi Calculi.

The approach used here to encode Turing Machines into intensional process calculi
can also be used to inform on similar approaches into non-intensional process
calculi such as $\pi$-calculus. Although faithfulness is lost, the
simplicity of intensional calculi in both: matching many names in a single interaction,
and of binding complex structures to a single name, becomes clearer when observing
the complexity required to use this approach in $\pi$-calculus.
Thus despite $\pi$-calculi only losing faithfulness directly, the complexity of
the encodings into $\pi$-calculi and having to work with many restrictions and
renamings highlights the elegance of the
encodings into intensional calculi.

\subsection*{Future Work}

The r\^ole of intensionality in process calculi has not been explored in depth
outside of particular calculi. A more general exploration of the expressiveness
of intensionality remains to be published.

The approach of encoding Turing Machines by directly capturing the state and
transitions has some similarities to the encoding of $SF$-logic \cite{jay2011}
into CPC \cite{GivenWilsonPhD}.
Adapting these approaches to other types of Turing Machines or rewriting
systems is also of interest.

\bibliographystyle{eptcs}
\bibliography{auto}

\end{document}